\documentclass[conference]{IEEEtran}
\IEEEoverridecommandlockouts
\usepackage{cite}
\usepackage{amsmath,amssymb,amsfonts}
\usepackage{algorithm}
\usepackage{algpseudocode}
\usepackage{graphicx}
\usepackage{textcomp}
\usepackage{url}
\usepackage{amsthm}
\usepackage{subcaption}
\theoremstyle{definition}
\newtheorem{dfn}{Definition}

\newtheorem{thm}[dfn]{Theorem}

\def\BibTeX{{\rm B\kern-.05em{\sc i\kern-.025em b}\kern-.08em
    T\kern-.1667em\lower.7ex\hbox{E}\kern-.125emX}}
\begin{document}

\title{A Fully Local Last-Generated Rule \\in a Blockchain
}

\author{\IEEEauthorblockN{Akira Sakurai}
\IEEEauthorblockA{\textit{Kyoto University} \\
Kyoto, Japan}
\and
\IEEEauthorblockN{Kazuyuki Shudo}
\IEEEauthorblockA{\textit{Kyoto University} \\
Kyoto, Japan}
}

\maketitle

\begin{abstract}
An effective method for suppressing intentional forks in a blockchain is the last-generated rule, which selects the most recent chain as the main chain in the event of a chain tie. This rule helps invalidate blocks that are withheld by adversaries for a certain period.
However, existing last-generated rules face an issue in that their applications to the system are not fully localized. In conservative cryptocurrency systems such as Bitcoin, it is desirable for methods to be applied in a fully local manner. 
In this paper, we propose a locally applicable last-generated rule. Our method is straightforward and is based on a relative time reference.
By conservatively setting the upper bound for the clock skews $\Delta_{O_i}$ to 200 s, our proposed method reduces the proportion $\gamma$ of honest miners following the attacker during chain ties by more than 40\% compared to existing local methods.
\end{abstract}


\section{Introduction}
Blockchain is the foundational technology used in decentralized currency systems, including Bitcoin \cite{bitcoin}. Blockchain systems are generally categorized into those that utilize proof-of-work and those that employ proof-of-stake. The focus of this study was on blockchains based on proof-of-work.

The security of a blockchain system is supported primarily by its incentive mechanisms. Specifically, the system is designed such that the most profitable strategy for each miner is to extend the longest chain. However, attacks that intentionally fork the chain to increase the adversary’s rewards have been identified \cite{majorityisnotenough} \cite{optimalSelfishMining} \cite{stubbornMining} \cite{Forkafterwithholding}. Such attacks undermine the system's consistency and lead to undesirable centralization of miners.

Various methods have been proposed to prevent intentional forks \cite{majorityisnotenough} \cite{oneweirdtricktostopselfishminers} \cite{PreventingBitcoinSelfishMiningUsingTransactionCreationTime} \cite{PublishorPerish} \cite{ZeroBlock} \cite{CounteringSelfishMininginBlockchains} \cite{PreventingSelfishMininginPublicBlockchainUsing} \cite{TFTstrategy} \cite{sakurai2024tiebreaking}\cite{Fortis}. One such method, addressed in this study, is the \textit{last-generated rule}. This is a tie-breaking rule that selects the most recently generated chain as the main chain when chains are in a chain tie. Here, a tie-breaking rule refers to a rule that determines the main chain when fork choice rules, such as the longest chain rule or GHOST \cite{GHOST}, cannot provide a unique decision. By applying the last-generated rule, it becomes possible to distinguish blocks held by adversaries for a certain period from others, effectively invalidating the adversary’s blocks. Unlike other countermeasures, the last-generated rule can be applied without requiring system-wide updates, such as hard forks or soft forks, or the need for strong synchrony, thereby making its implementation easy.

However, existing last-generated rules suffer from a limitation in that their application is not localized to each miner. These methods require trusted third parties \cite{oneweirdtricktostopselfishminers} or sharing new messages \cite{sakurai2024tiebreaking}. Conservative operations are preferred in currency systems such as Bitcoin. In such systems, it is desirable that new methods be localized.

In this paper, we propose a simple last-generated rule that can be applied in a fully local manner. Here, ``fully local'' means that no additional communication is required. Whereas existing last-generated rules rely on an absolute time standard, our method is based on a relative time standard. Specifically, our method assumes that there is a known upper bound for the clock skews between miners. On the basis of this assumption, the proposed last-generated rule rejects blocks with significant discrepancies between the local clock and the block's generation time instead of selecting the newer block. Furthermore, we demonstrate the superiority of our method over local methods through both theoretical analysis and simulation experiments.

\section{Related Work}
Here, we examine the existing last-generated rules, focusing on whether they are local or whether they function as last-generated rules.

Eyal et al. proposed Selfish Mining and, as a countermeasure, introduced the random rule, a tie-breaking rule that randomly determines the main chain during chain ties~\cite{oneweirdtricktostopselfishminers}. Under the random rule, the hashrate proportion $\gamma$ of honest miners who follow the adversary's chain during chain ties is 0.5. Their method is completely local and was actually implemented in Ethereum, which had employed proof-of-work at the time.

To further suppress $\gamma$ beyond the level achieved by the random rule, Heilman proposed a last-generated rule, which selects the most recently generated chain during chain ties~\cite{oneweirdtricktostopselfishminers}. The last-generated rule is more powerful than the random rule and reduces $\gamma$ closer to 0. However, his method requires a trusted third party, making it non-local. This dependency is a significant drawback for trustless blockchain systems.

Lee et al. proposed a last-generated rule that uses the transaction creation time as a new time standard \cite{PreventingBitcoinSelfishMiningUsingTransactionCreationTime}. However, this method has two major limitations. First, it requires embedding creation times in transactions, increasing the transaction size, and imposing a significant burden on both users and miners. Second, the time standard can be easily manipulated by adversaries, who can forge the transaction creation times, which is a critical flaw. Saad et al. and Reno et al. proposed similar last-generated rules that suffer from the same drawbacks \cite{CounteringSelfishMininginBlockchains} \cite{PreventingSelfishMininginPublicBlockchainUsing}.

Bi\c{c}er et al. proposed a method that combines the random rule with Heilman's last-generated rule, which does not rely on a trusted third party \cite{Fortis}. Specifically, this method probabilistically prioritizes the block with the more recent timestamp in the case of a chain tie, selecting it as part of the main chain. However, a significant drawback of this method is that it does not adequately consider the fact that an attacker can easily set a timestamp far in the future. If an attacker does this, the effectiveness of the method drops below that of the random rule.

Sakurai et al. proposed a last-generated rule that uses partial proof-of-work as a time standard \cite{sakurai2024tiebreaking}. Here, partial proof-of-work refers to block headers that, although weaker than full blocks, still possess a certain amount of proof-of-work. Their method requires the network to share messages containing a small amount of proof-of-work, which makes it non-local. This presents a practical challenge.

\section{Proposed Method}\label{proposedmethod}
The proposed method is a last-generated rule based on the assumption that there exists a known upper bound for the clock skews between miners. Under this assumption, when an honest miner successfully generates a block, miners can expect a certain level of accuracy in the timestamp included in the block. In other words, if a block's timestamp deviates sufficiently from its own clock, each miner can determine whthere the block was generated by an adversary. 

The proposed method is locally applicable, requiring only that each miner sets an upper bound for the block propagation time $\Delta_{B_i}$ and an upper bound for the clock skews between miners $\Delta_{O_i}$. These parameters can be configured locally, without requiring coordination with or reliance on other miners.

\subsection{Model}\label{model}
We make two assumptions regarding synchrony. First, we assume the existence of a known upper bound $\Delta_B$ for the block propagation time. Second, we assume the existence of a known upper bound $\Delta_O$ for the clock skews between miners. By definition, $\Delta_O$ satisfies
\begin{equation}
    \max_{i, j \in V} |CL_i - CL_j| \leq \Delta_O 
\end{equation}
where $V$ represents the set of miners, and $CL_i$ denotes the local clock time of miner $i$.

Next, we make several assumptions about the adversary. We assume that the adversary can immediately detect blocks generated by honest miners. Moreover, the adversary can transmit blocks to all miners with no time delay.

It is noteworthy that these assumptions do not hold for real-world networks. Specifically, the known upper bounds $\Delta_B$ and $\Delta_O$ are unknown. In Sections~\ref{proposedmethod} and~\ref{theoreticeval}, we describe the construction and evaluation of the proposed method using this model. However, in Section \ref{simulation}, we demonstrate the effectiveness of the method even in networks where these assumptions do not hold.

\subsection{Details}
Before discussing the details of the proposed method, we present the following fundamental theorem.
\begin{thm}\label{theorem1}
Let $t_B$ denote the timestamp of a block, and let $a_l$ represent the reception time of that block. If the following condition holds, then the creator of the block is an adversary:
\begin{equation}
    a_l - t_B < -\Delta_O \lor a_l - t_B > \Delta_O + \Delta_B. \label{corecondition}
\end{equation}
    
\end{thm}

\begin{proof}
We demonstrate the contrapositive. When an honest miner generates a block, the minimum value of $a_l - t_B$ occurs when the block creator's clock is $\Delta_O$ ahead of the receiving miner's clock, and the block is received simultaneously upon generation. In this case, $a_l - t_B$ equals $-\Delta_O$. Next, the maximum value of $a_l - t_B$ occurs when the receiving miner's clock is $\Delta_O$ ahead of the block creator's clock and the block is received $\Delta_B$ after its generation. In this scenario, $a_l - t_B$ equals $\Delta_O + \Delta_B$. Therefore, the following inequality holds:
\begin{equation}
-\Delta_O \leq a_l - t_B \leq \Delta_O + \Delta_B.
\end{equation}
\end{proof}

\begin{algorithm}
\caption{Fork choice rule using the proposed method for miner $i$}
\label{algousingproposedmethod}
\begin{algorithmic}[1]
\Procedure{getMainChain}{$\textit{C}_\textit{1}, \ldots, \textit{C}_\textit{k}$}
    \State $\textit{candidates} \gets \text{forkChoiceRule}(\textit{C}_\textit{1}, \ldots, \textit{C}_\textit{k})$
    
    \State $\textit{earliestTime} \gets \textit{C}_\textit{1}.\textit{arrivalTime}$
    \For{$\textit{C} \in \textit{candidates}$}
        \State $\textit{earliestTime} \gets \text{min}(\textit{earliestTime}, \textit{C}.\textit{arrivalTime})$
    \EndFor

    \State $\textit{prospects} \gets \emptyset$
    \While{$\exists C \in \textit{candidates} \text{ s.t. } C.\textit{arrivalTime} - \textit{earliestTime} > \textit{w}$}
      \State $\textit{prospects} \gets \textit{prospects} \cup \{C\}$
    \EndWhile
    
    \While{$(\exists C \in \textit{prospects} \text{ s.t. } (C.\textit{arrivalTime} - C.\textit{timestamp} < - \Delta_{O_i}) \lor (C.\textit{arrivalTime} - C.\textit{timestamp} >  \Delta_{O_i} + \Delta_{B_i})) \land |\textit{prospects}| \neq 0$}
      \State $\textit{prospects} \gets \textit{prospects} \setminus \{C\}$
    \EndWhile
    \If{$|\textit{prospects}| = 0$}
        \State Randomly choose the \textit{mainchain} from \textit{candidates}.
        \State \Return \textit{mainchain}
    \EndIf

    \State Randomly choose the \textit{mainchain} from \textit{prospects}.
    \State \Return \textit{mainchain}
\EndProcedure 
\end{algorithmic}
\end{algorithm}

Building on the aforementioned model and theorem, we present the proposed method. Algorithm~\ref{algousingproposedmethod} outlines the fork choice rule that uses the proposed method. Function \texttt{getMainChain} accepts an arbitrary number of chains as arguments and returns the main chain. The proposed method is a last-generated rule and, to prevent intentional chain ties against pre-generated blocks by honest miners imposes an acceptance window $w = \Delta_{B_i}$ for chains in a tie. By configuring the acceptance window $w$ in this manner, the impact of such attacks is sufficiently suppressed \cite{sakurai2024tiebreaking}. After restricting the chains in a tie in this manner (lines 3--10), each chain is evaluated to determine whether Condition \ref{corecondition} is satisfied. If Condition \ref{corecondition} is not satisfied for a chain, that chain is excluded from the main chain candidates (lines 11--13). If no candidates remain after exclusion, the main chain is randomly selected from the candidates prior to exclusion on the basis of Condition \ref{corecondition} (lines 14--17). Otherwise, the main chain is randomly selected from the remaining main chain candidates (lines 18 and 19).

Here, $\Delta_{B_i}$ represents the upper bound for the block propagation time and $\Delta_{O_i}$ denotes the upper bound for the clock skews, both independently set by each miner. Importantly, each miner can independently determine these values without coordinating with other miners. The settings for $\Delta_{O_i}$ and $\Delta_{B_i}$ can vary among miners. Furthermore, the proposed method only requires each miner to set $\Delta_{O_i}$ and $\Delta_{B_i}$. In this sense, the proposed method is locally applicable compared to the existing last-generated rules.

\section{Evaluation}
We evaluated the proposed method through both theoretical analysis and simulation experiments. Specifically, we focused on the total hashrate ratio $\gamma$ of honest miners who follow an adversary's chain when the adversary intentionally causes chain ties. We primarily compared the proposed method with the random rule, which is currently the most effective fully local method aside from the proposed one.

\subsection{Theoretical Analysis}\label{theoreticeval}
We conduct a theoretical analysis of the proposed method. For the theoretical analysis, we assumed that the assumptions outlined in Section\ref{model} hold true.
Several additional assumptions were made for the ease of analysis. First, Each miner sets $\Delta_{O_i} = \Delta_O$ and $\Delta_{B_i} = \Delta_B$. Next, we also assumed that the adversary is aware of the values of $\Delta_O$ and $\Delta_B$ used by the honest miners.

Under these assumptions, the following theorem holds.
\begin{thm}\label{theorem2}
On average, $\gamma$ satisfies
\begin{align}
  \gamma \leq \frac{1}{2} - \frac{1}{2}\exp\left(-\frac{2 \Delta_O + 3 \Delta_B}{T}\right).
\end{align}
where $T$ is the average block generation interval.
\end{thm}
\begin{proof}
  Consider a scenario where the adversary generates a block and attempts to prevent an honest miner from meeting Condition~\ref{corecondition}. First, assume that another honest miner generates a block simultaneously with the attacker. After $2\Delta_B$, the acceptance window for the target honest miner will have closed. This is because it takes $\Delta_B$ for the block to reach the target honest miner, and an additional $w = \Delta_B$ for the acceptance window to end. Thus, the attacker must ensure that Condition~\ref{corecondition} is not met by the honest miner within $2\Delta_B$ after block generation. 

  Then, the attacker sets the block's timestamp to be $\Delta_O + 2 \Delta_B$ ahead of the target miner’s clock. Under these conditions, $2\Delta_O + 3\Delta_B$ after generating the block, the attacker’s block will satisfy Condition~\ref{corecondition}.

  As a result, the probability ($P_{win}$) that the honest miner’s block will become part of the main chain is:
\begin{align}
    P_{win} \geq &\frac{1}{2} \left(1 - \exp\left(-\frac{2 \Delta_O + 3 \Delta_B}{T}\right)\right) + \exp\left(-\frac{2 \Delta_O + 3 \Delta_B}{T}\right) \label{probhonestwin} \\
               =& \frac{1}{2} + \frac{1}{2} \exp\left(-\frac{2 \Delta_O + 3 \Delta_B}{T}\right)
\end{align}
  The left term on the left side of Inequality \ref{probhonestwin} represents the probability that an honest miner generates a block within $2\Delta_O + 3\Delta_B$ after the adversary generates a block, possibly causing all honest miners to choose the main chain randomly. The right-hand term represents the probability that an honest miner generates a block after $2\Delta_O + 3\Delta_B$ following the adversary's block generation, leading to all honest miners following and mining on the honest miner's block.

  Therefore, $\gamma$ satisfies
\begin{align}
  \gamma  &= 1 - P_{win} \\
          &\leq 1 - \left(\frac{1}{2} + \frac{1}{2}\exp\left(-\frac{2 \Delta_O + 3 \Delta_B}{T}\right)\right) \\
          &= \frac{1}{2} - \frac{1}{2}\exp\left(-\frac{2 \Delta_O + 3 \Delta_B}{T}\right).
\end{align}
\end{proof}

From Theorem~\ref{theorem2}, we can observe that the proposed method effectively suppresses intentional chain ties by adversaries compared with the random rule($\gamma = 1/2$), while still maintaining a high degree of locality. 

\subsection{Simulation Experiments}\label{simulation} 
Simulations were conducted to evaluate the proposed method. In the theoretical analysis, several assumptions were made that may not hold in real-world scenarios. For instance, whereas the theoretical analysis assumed no clock skews, in large-scale distributed systems such as Bitcoin, such skews among miners may have a significant impact. In this section, we assessed the performance of the proposed method in a more realistic and quantitative manner considering factors such as clock skews.

\subsubsection{Simulation Settings}

The simulator used in this study was event-driven and inspired by the blockchain network simulator SimBlock \cite{simblock}. The average block generation interval was set to 600 s. There were 1,000 honest miners, each with equal hashrate. Additionally, there was a single adversary whose hashrate was equal to the combined hashrate of all honest miners. We assumed that the block propagation times for both the adversary and honest miners were zero, which means that we did not consider the impact of forks by honest miners.

When the adversary successfully generated a block, it did not immediately publish but withheld it. Instead, when an honest miner generated a block, the adversary released its block to cause a chain tie. When the adversary successfully generated consecutive blocks, the height difference between the chains maintained by the honest miners and the adversary became two blocks. In this case, the adversary published its block as well.

The four primary parameters considered in this experiment were 
\begin{itemize} 
    \item The upper bound for the block propagation time set by each miner, $\Delta_{B_i}$ 
    \item The upper bound for the clock skews between miners set by each miner, $\Delta_{O_i}$ 
    \item The clock offset of each miner, $O_i$ 
    \item The clock offset of the adversary, $a_t$ 
\end{itemize} 
The clock offset $O_i$ for each miner followed a normal distribution with a mean of 0. In addition, considering the current statistics on Bitcoin's block propagation time and fork rates, the upper bound for the block propagation time $\Delta_{B_i}$ was conservatively set to 20 s \cite{statsofbitcoin}.

Next, we set $\Delta_{O_i}$. In blockchain networks that employ proof-of-work, it is difficult to estimate the upper bound for the clock skews $\Delta_{O}$ between nodes. Therefore, we examined Ethereum \cite{ethereum}, a similar large-scale blockchain network. Currently, Ethereum uses a slot-based proof-of-stake system \cite{buterin2020combiningghostcasper}. In Ethereum, the effect of clock skews between nodes manifests as forks in the blockchain. That means that it is possible to estimate the upper bound for the clock skews $\Delta_{O}$ between the nodes by observing the number of forks. In September 2024, the Ethereum fork rate was approximately 0.3\% \cite{etherscan_blocks_forked}. Given that each slot in Ethereum is 12 s, this suggests that the upper bound for the clock skews among most block producers is within a few seconds. Therefore, in the simulation experiments conducted in this study, we initially set $\Delta_{O_i}$ to 20 s. Subsequently, we conservatively increased $\Delta_{O_i}$ to 200 s.

\subsubsection{Simulation Result}
\begin{figure}[t]
    \centering
    \includegraphics[width=1\linewidth]{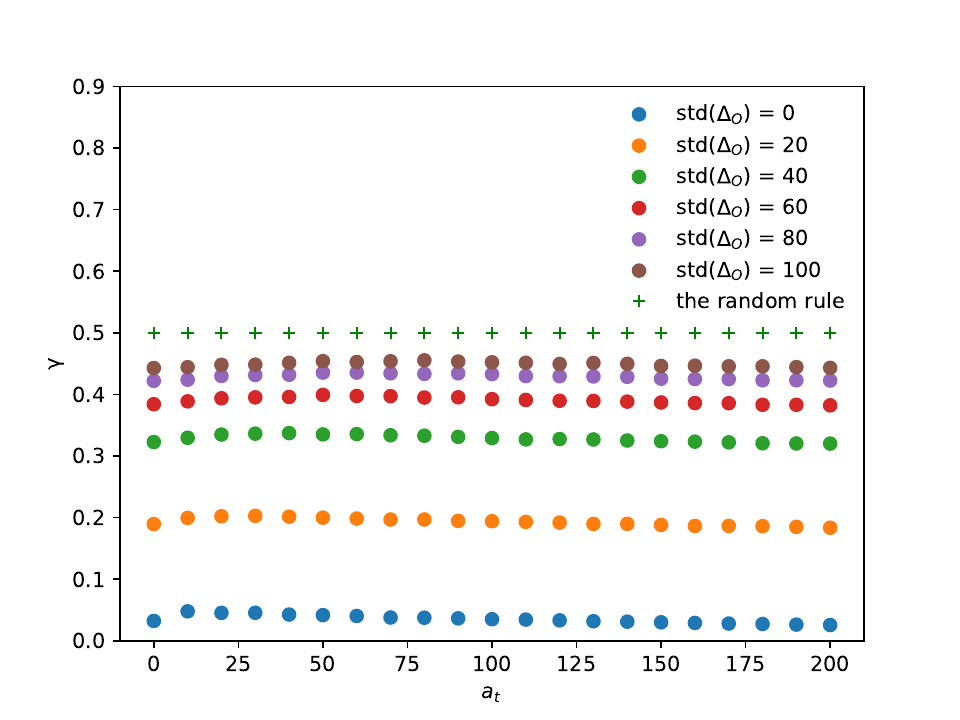}
    \caption{Simulation results for $\gamma$ when $\Delta_{O_i} = 20$ s.}
    \label{doi20}
\end{figure}
\begin{figure}[t]
    \centering
    \includegraphics[width=1\linewidth]{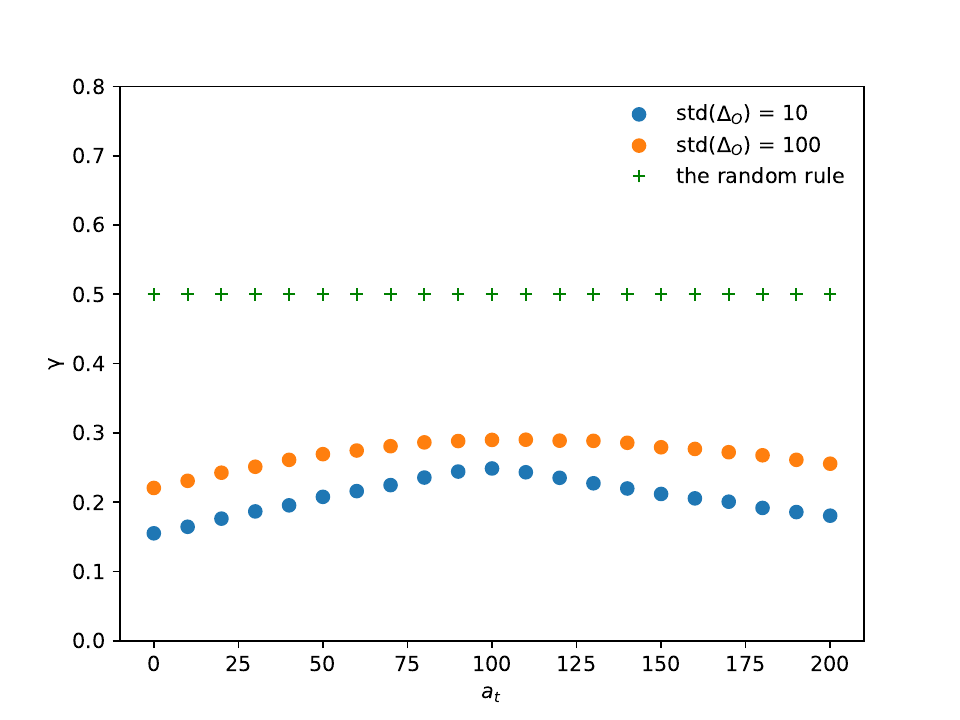}
    \caption{Simulation results for $\gamma$ when $\Delta_{O_i} = 200$ s.}
    \label{doi200}
\end{figure}

Figure~\ref{doi20} illustrates $\gamma$ when $\Delta_{O_i}$ was set to 20 s. $\gamma$ increased with the standard deviation of $\Delta_{O_i}$ ($ std(\Delta_{O_i})$) because, with a higher $ std(\Delta_{O_i})$, blocks generated by honest miners are less likely to satisfy Condition \ref{corecondition}. However, in all scenarios, the proposed method outperformed the random rule, demonstrating its superiority among fully local methods.

Figure~\ref{doi200} shows $\gamma$ when $\Delta_{O_i}$ was set to 200 s. By conservatively setting $\Delta_{O_i}$ to a larger value, the effectiveness of the proposed method decreased when $ std(\Delta_{O_i})$ was small. However, this setting helped maintain the effectiveness of the proposed method even when $ std(\Delta_{O_i})$ was large. Unlike $\Delta_B$, $\Delta_O$ is difficult to predict in actual networks. The simulation results suggest that when conservatively applying the proposed method, it is desirable to set $\Delta_{O_i}$ to a larger value, such as 200 s.

\section{Conclusion} 
We proposed a fully local last-generated rule. Our approach is based on a relative time standard. Through a theoretical analysis and simulation experiments, we confirmed that the proposed method functioned effectively as a fully local last-generated rule.

\bibliography{hoge} 
\bibliographystyle{IEEEtran.bst}

\end{document}